\documentclass{llncs}
\usepackage{algorithm}
\usepackage{algorithmicx}
\usepackage{algpseudocode}
\usepackage{amsmath}
\usepackage{graphics}
\usepackage{graphicx}
\usepackage{makeidx}  
\usepackage{multicol}
\usepackage{subfigure}
\usepackage{tikz} 
\usepackage{url}

\begin{document}
\frontmatter          
\pagestyle{headings}  
\addtocmark{Hamiltonian Mechanics} 
\mainmatter              

\newcommand{\LCSK}{$LCS_k$}
\newcommand{\LCSKP}{$LCS_{k+}$}

\newcommand{\fortilde}[1]{\urlstyle{same}\url{#1}}

\title{Fast and simple algorithms for computing both \LCSK{} and \LCSKP{}}
\titlerunning{Fast \LCSK{} and \LCSKP{} }  
%

\author{Filip Paveti\'{c}\inst{1}, Ivan Katani\'{c}\inst{2}, Gustav Matula\inst{2}, Goran \v{Z}u\v{z}i\'{c}\inst{3} \and Mile \v{S}iki\'{c}\inst{2}}

\authorrunning{Filip Paveti\'{c} et al.} 
%
%
\institute{
Google Switzerland GmbH, Z\"urich, Switzerland\\
  \email{fpavetic@google.com}\\
\and
Faculty of Electrical Engineering and Computing, University of Zagreb, Croatia\\
  \email{\{ivan.katanic, gustav.matula, mile.sikic\}@fer.hr}\\
\and
Carnegie Mellon University, Pittsburgh, USA\\
  \email{gzuzic@cs.cmu.edu}
}

\maketitle              

\begin{abstract}
Longest Common Subsequence ($LCS$) deals with the problem of measuring similarity of two strings. While this problem has been analyzed for decades, the recent interest stems from a practical observation that considering single characters is often too simplistic. Therefore, recent works introduce the variants of $LCS$ based on shared substrings of length exactly or at least $k$ (\LCSK{} and \LCSKP{} respectively). The main drawback of the  state-of-the-art algorithms for computing \LCSK{} and \LCSKP{} is that they work well only in a limited setting: they either solve the average case well while being suboptimal in the pathological situations or they achieve a good worst-case performance, but fail to exploit the input data properties to speed up the computation. Furthermore, these algorithms are based on non-trivial data structures which is not ideal from a practitioner's point of view. We present a single algorithm to compute both \LCSK{} and \LCSKP{} which outperforms the state-of-the art algorithms in terms of runtime complexity and requires only basic data structures. In addition, we implement an algorithm to reconstruct the solution which offers significant improvement in terms of memory consumption. Our empirical validation shows that we save several orders of magnitude of memory on human genome data. The C++ implementation of our algorithms is made available at: \url{https://github.com/google/fast-simple-lcsk}.

\keywords{longest common subsequence, string similarity, efficient dynamic programming, bioinformatics, memory optimization}
\end{abstract}
%
%

\section{Introduction}

Measuring the similarity of strings is one of the fundamental problems in computer science. It is very useful in many real-world applications such as DNA sequence comparison \cite{lcsk}, differential file analysis and plagiarism detection \cite{chen2010plagiarism}. One of the most popular techniques for efficient measurement of string similarity is the Longest Common Subsequence (\emph{LCS}) \cite{lcs.hunt,Kuo:1989:IAF:74697.74702,lcsfragments}.\\

\noindent\textbf{LCS and extensions.} Recently, there has been some critique of $LCS$ being an oversimplified way to measure string similarity as it does not distinguish well between the sequences which consist mainly of consecutive characters and the ones which do not \cite{benson2016lcsk,UekiDKMNYBIS16}. To overcome this limitation, extensions of $LCS$ to more general variants have been proposed (see Example \ref{lcsk.example}). In particular, Benson et al. \cite{lcsk,benson2016lcsk} suggested \LCSK, which computes the similarity by counting the number of non-overlapping substrings of length $k$ contained in both strings. Another extension was given by Paveti\'{c} et al. \cite{PaveticZS14}: the $LCS_{k++}$ computes similarity by summing the lengths of non-overlapping substrings of length \textbf{at least} $k$ contained in both strings. These variants have been applied in bioinformatics \cite{Sovic020719}. Later it has been renamed as $LCS_{\ge k}$ by Benson et al. \cite{benson2016lcsk} and as (used in this paper) \LCSKP{} by Ueki et al. \cite{UekiDKMNYBIS16}\footnote{Different authors used different names, however the definitions are the same.}.

\begin{example}[Values of \LCSK{} and \LCSKP{} for various string pairs]
\begin{itemize}
  \item $LCS_{3}(ABCBA, ABCBA) = 1$ $(ABC)$ or $(BCB)$ or $(CBA)$
  \item $LCS_{3+}(ABCBA, ABCBA) = 5$ $(ABCBA)$
  \item $LCS_{2}(ABXXXCDE, ABYYYCDE) = 2$ $(AB, CD)$ or $(AB, DE)$
  \item $LCS_{2+}(ABXXXCDE, ABYYYCDE) = 5$ $(AB, CDE)$
  \item $LCS_{1}(AAA, AA) = LCS(AAA, AA) = 2$ $(A, A)$
  \item $LCS_{1+}(AAA, AA) = LCS(AAA, AA) = 2$ $(A, A)$
\end{itemize}
\label{lcsk.example}
\end{example}

\noindent\textbf{State of the art.} For $LCS_k$ and $LCS_{k+}$ to be useful in practice, we need to be able to compute them efficiently. State-of-the-art algorithms are often parametrized by the total number of matching $k$-length substring pairs between the input strings, denoted by $r$. The observation that $r$ is often limited in real-world data can be used to speed up the computation. Existing algorithms usually specialize for the situations when $r$ is either low or high. Deorowicz and Grabowski \cite{lcskfast} proposed several algorithms for efficient computation of $LCS_k$. Most notably, their \textit{Sparse} algorithm allows both the computation of \LCSK{} and its reconstruction in $\mathcal{O}(m + n + r\log l)$ time and $\mathcal{O}(r)$ memory complexity, where $l$ is the length of the optimal solution and $m,n$ are the lengths of the two input strings. In their approach they adapt the Hunt-Szymanski \cite{lcs.hunt} paradigm in a way that requires them to use a persistent red-black tree. Paveti\'{c} et al. \cite{PaveticZS14,github:lcskpp} proposed an algorithm based on the Fenwick tree data structure \cite{fenwick1994new} to efficiently compute $LCS_{k+}$ in $\mathcal{O}(m + n + r\log r)$. Ueki et al. \cite{UekiDKMNYBIS16} proposed an algorithm which achieves a better worst-case complexity than Paveti\'{c} and \v{Z}u\v{z}i\'{c} \cite{github:lcskpp} when $r\sim mn$, but a worse performance in the situations when $r$ is small. The complexity of their algorithm is $\mathcal{O}(mn)$ and does not depend on $r$.\\

\noindent\textbf{Improving the current state.} The algorithms to compute \LCSK{} and \LCSKP{} are either of unsatisfactory runtime complexity or they rely on using complex data structures. Implementing these complex algorithms is extremely time consuming for a practitioner who usually has to make an experiment-based decision about which similarity measure is even useful for their cause. Therefore, it is valuable to have a simple and fast algorithm to compute both \LCSK{} and \LCSKP{}. Our contributions in this paper are:
\begin{itemize}
\item[$\bullet$] We propose an algorithm to compute \LCSK{} which achieves a favorable runtime complexity when compared to the previously known algorithms. The runtime complexity of the algorithm is $\mathcal{O}(\min(r\log l, r + ml))$, where $l$ is the length of the optimal solution. (Section \ref{lcsk.algorithm})
\item[$\bullet$] We show that the same algorithm can be easily extended to compute $LCS_{k+}$. This unifies the solutions for both of these problems. (Section \ref{section:lcskpp})
\item[$\bullet$] We propose a heuristic to reduce memory needed to reconstruct the solution. Experiments on the human genome demonstrate that it reduces the memory usage by several orders of magnitude. (Section \ref{section:reconstructionmemory})
\item[$\bullet$] Our algorithms do not rely on complex data structures such as Fenwick or a persistent red-black tree.
\item[$\bullet$] To speed up future research on the topic, we made the implementation of our algorithms available at \url{https://github.com/google/fast-simple-lcsk}. As far as we know, this is the first widely accessible implementation of algorithms to compute \LCSK{} and \LCSKP{}.
\end{itemize}

\section{Preliminaries}
\label{preliminaries}

This section contains the definitions useful throughout the rest of the paper. Even though we inherit some of the definitions from other sources, we state them here for the sake of completeness. The inputs to all the algorithms are strings $A$ of length $m$ and $B$ of length $n$, both over alphabet $\Sigma$. Without loss of generality we can assume that $m \le n$. We use $X[i:j)$ to denote the substring of $X$, starting at (inclusive) index $i$ and ending at (exclusive) index $j$. Using that notation it holds that $A[0:m)=A$ and $B[0:n)=B$.

\begin{definition}[The $LCS_k$ problem \cite{benson2016lcsk}]
Given two strings $A$ and $B$ of length $m$ and $n$, respectively, and an integer $k \ge 1$, we say that $C$ is a common subsequence in \emph{exactly} $k$ length substrings of $A$ and $B$, if there exist $i_1, ... , i_t$ and $j_1, ... , j_t$ such that $A[i_s:i_s+k) = B[j_s:j_s+k) = C[p_s:p_s+k)$ for $1 \le s \le t$, and $i_s + k \le i_{s+1}, j_s + k \le j_{s+1}$ and $p_{s+1} = p_s + k$ for $1 \le s < t, p_1 = 0$ and $|C| = p_t + k$. The longest common subsequence in \emph{exactly} $k$ length substrings ($LCS_{k}$) equals to maximum possible $t$, such that the mentioned conditions are met.
\end{definition}

\begin{definition}[The $LCS_{k+}$ problem \cite{PaveticZS14,benson2016lcsk,UekiDKMNYBIS16}]
Given two strings $A$ and $B$ of length $m$ and $n$, respectively, and an integer $k \ge 1$, we say that $C$ is a common subsequence in \emph{at least} $k$ length substrings of $A$ and $B$, if there exist $i_1, ... , i_t$, $j_1, ... , j_t$ and $l_1, ... , l_t$ such that $A[i_s:i_s+l_s) = B[j_s:j_s+l_s) = C[p_s:p_s+l_s)$ and $l_s \ge k$ for $1 \le s \le t$, and $i_s + l_s \le i_{s+1}, j_s + l_s \le j_{s+1}$ and $p_{s+1} = p_s + l_s$ for $1 \le s < t, p_1 = 0$ and $|C| = p_t + l_t$. The longest common subsequence in \emph{at least} $k$ length substrings ($LCS_{k+}$) equals to maximum possible sum $\sum_{i=1}^{t}l_i$, such that the mentioned conditions are met.
\end{definition}

\begin{figure}[t]
\centering

\begin{tikzpicture}[scale=0.7]
\draw (0,0) grid (10,5);
\node at (0.5, 5.5) {C};
\node at (1.5, 5.5) {T};
\node at (2.5, 5.5) {A};
\node at (3.5, 5.5) {T};
\node at (4.5, 5.5) {A};
\node at (5.5, 5.5) {G};
\node at (6.5, 5.5) {A};
\node at (7.5, 5.5) {G};
\node at (8.5, 5.5) {T};
\node at (9.5, 5.5) {A};

\node at (-0.5, 4.5) {A};
\node at (-0.5, 3.5) {T};
\node at (-0.5, 2.5) {T};
\node at (-0.5, 1.5) {A};
\node at (-0.5, 0.5) {T};

\node at (1.5, 2.5) {a}; 
\draw (1.5,2.5) circle [radius=0.25]; 
\node at (2.25, 1.75) {a}; 
\draw (2.25,1.75) +(-0.25, -0.25) rectangle +(0.25, 0.25);

\node at (2.75, 1.25) {b};
\draw (2.75,1.25) circle [radius=0.25];
\node at (3.5, 0.5) {b};
\draw (3.5,0.5) +(-0.25, -0.25) rectangle +(0.25, 0.25);

\node at (2.5, 4.5) {c};
\draw (2.5,4.5) circle [radius=0.25];
\node at (3.5, 3.5) {c};
\draw (3.5,3.5) +(-0.25, -0.25) rectangle +(0.25, 0.25);

\node at (3.5, 2.5) {d};
\draw (3.5,2.5) circle [radius=0.25];
\node at (4.5, 1.5) {d};
\draw (4.5,1.5) +(-0.25, -0.25) rectangle +(0.25, 0.25);

\node at (8.5, 2.5) {e};
\draw (8.5,2.5) circle [radius=0.25];
\node at (9.5, 1.5) {e};
\draw (9.5,1.5) +(-0.25, -0.25) rectangle +(0.25, 0.25);
\end{tikzpicture}

\caption{The Figure shows the start and end points of the match pairs produced by the two strings. In this example strings $A=ATTAT$ and $B=CTATAGAGTA$ construct exactly five match pairs for $k=2$, denoted with $a$ to $e$. Start points of the pairs are represented by circles and their end points are represented by squares. The following holds: ``$c$ precedes $e$'', while the following does not hold: ``$a$ precedes $b$'', ``$c$ precedes $d$'' (Definition \ref{precedence}). Note that a start point of one match pair can share coordinates with the end point of another: e. g. end point of $a$ and start point of $b$.}
\label{fig:kmatchpairs}
\end{figure}

\noindent We state the recurrence relation to compute \LCSK{}($i,j$)=\LCSK{}($A[0:i),B[0,j)$) for two strings $A$ and $B$, given in \cite{lcsk}:

\begin{eqnarray}
\label{relation.lcsk}
  LCS_k(i,j) & = & \max
  \left \{
    \begin{array}{ll}
      LCS_k(i-1, j) & $ if $ i \ge 1\\
      LCS_k(i, j-1) & $ if $ j \ge 1\\
      LCS_k(i-k, j-k)+1 & $ if $ A[i-k:i)=B[j-k:j)\\
    \end{array}
  \right .
\end{eqnarray}
The choice to add $1$ or $k$ in the last line of Equation \ref{relation.lcsk} is arbitrary since it influences the final result only by a constant factor. We choose to add $+1$ to be consistent with prior definitions of \LCSK. When expanding the relation to \LCSKP, we need to adjust it as shown in \cite{PaveticZS14,UekiDKMNYBIS16}:
\begin{eqnarray}
\label{relation.lcskp}
  LCS_{k+}(i,j) & = & \max
  \left \{
    \begin{array}{ll}
      LCS_{k+}(i-1, j) & i \ge 1\\
      LCS_{k+}(i, j-1) & j \ge 1\\
      \text{$\max\limits_{\substack{k\le k' \le min(i,j)\\A[i-k':i)=B[j-k':j)}}$} LCS_{k+}(i-k',j-k')+k' & \\
    \end{array}
  \right .
\end{eqnarray}

\begin{definition}[Match pair \cite{lcsk}]
Given the strings $A$, $B$ and integer $k \ge 1$ we say that at $(i,j)$ there is a \textbf{match pair} if $A[i:i+k)=B[j:j+k)$. $(i,j)$ is also called the \textbf{start point} or the \textbf{start} of the match pair. $(i+k-1,j+k-1)$ is called the \textbf{end point} or the \textbf{end} of the match pair.
\end{definition}

\begin{definition}[Precedence of match pairs]
Let $P$=$(i_P, j_P)$ and $G$=$(i_G, j_G)$ be match pairs. Then $G$ \textbf{precedes} $P$ if $i_G+k \le i_P$ and $j_G+k \le j_P$. In other words, $G$ precedes $P$ if the end of G is on the upper left side of the start of P in the dynamic programming table (see Figure \ref{fig:kmatchpairs}).
\label{precedence}
\end{definition}

\subsection{Efficient algorithms for computing LCS}
The known efficient algorithms are based on the observation that in order to compute \emph{LCS} via a classic dynamic programming approach\footnote{This is usually done by filling out a matrix based on the relation which we get after we set $k=1$ in either of the Equation \ref{relation.lcsk} or Equation \ref{relation.lcskp}.}, it is not always necessary to fill out the entire matrix. If it happens that many entries repeat in the cases when two strings have only a few pairs of matching characters, it is possible to design a structure which stores the matrix in a compressed form. We sketch the main ideas behind these techniques since we later use them as building blocks.\\

\noindent The algorithm by Hunt and Szymanski \cite{lcs.hunt,Bergroth:2000:SLC:829519.830817} traverses only the matching character pairs of the two strings in row-major order. The main idea is to maintain an array $M$ such that $M_d$ holds the minimum $j$ such that there is some already processed row $i$ for which $LCS(i, j)=d$. For simplicity we define $M_0 = 0$ and $M_d = \infty$ if no such $j$ exists. In simpler terms, $M$ is a compressed representation of the dynamic programming table, storing only the boundaries of same-value intervals. This is obviously useful when a row contains many repeated values. Note that $M$ is non-decreasing. For every point $(i, j)$ corresponding to matching character pairs in row $i$, let us find the biggest $d$ such that $M_d < j$. Then there must exist a common subsequence of length $d$ ending with $(i', j')$ where $i' < i$ and $j' < j$. Such a subsequence can be extended by $(i, j)$, so we know that after processing all the points in the current row we will have $M_{d+1} \le j$. It will suffice to find a $d$ such that $M_d < j \le M_{d+1}$ (easily done using binary search), and set $M_{d+1}$ to $\min(M_{d+1}, j)$ (since we're extending a subsequence of length $d$ ending at column $M_d$ into a subsequence of length $d+1$ ending at column $j$). We note that the order in which the points of a row are processed matters: they should be ordered descending by column, so that the updates to $M$ with the results of a new row happen effectively at the same time (otherwise queries are influenced by previous updates from the same row, which leads to incorrect results).\\

\noindent Hunt's algorithm was modified by Kuo and Cross \cite{Kuo:1989:IAF:74697.74702} by replacing the binary search for each point in a row with a linear scan of $M$ together with all the points in a row. Specifically, as we process the points of the current row, we also maintain an index $d$ into $M$ which we increment until $M_{d+1} \ge j$ for the current point $(i, j)$. The increments of $d$ amortize over the length of $M$, so the total complexity is $\mathcal{O}(r_i + l)$, where $r_i$ is the number of points in row $i$ and $l$ is the length of the $LCS$. When $r_i \ll l$ this algorithm is  performing worse than Hunt's variant (which would have a runtime complexity of $\mathcal{O}(r_i \log r_i)$). However, it does becomes a significant improvement as $r_i$ approaches $l$.

\section{Algorithm to compute \LCSK{}}
\label{lcsk.algorithm}

In this section we show how to compute \LCSK{} efficiently. The main observation is that processing start and end points of the match pairs independently allows us to directly re-use the techniques for the efficient computation of $LCS$. Additionally, we dynamically adapt the computation between situations where the number of match pairs $r$ is low as well as high, in order to secure a good worst-case performance.  We achieve a runtime complexity of $\mathcal{O}(m + n + r + \min(r\log l, r + ml))$ and the memory complexity of $\mathcal{O}(l + m + n)$.

\subsection{Decoupling the starts and ends of the match pairs}
Another way to formulate the computation of \LCSK{} is to view it as a problem of finding the longest chain of match pairs. This formulation is an extension of the one previously used in the $LCS$ literature \cite{Goeman2002}. Given a match pair $P$, it is possible to compute $LCS_k(P)$ using the following relation:
\begin{eqnarray}
  LCS_k(P) & = &
  \left \{
    \begin{array}{ll}
      1 & \text{ if no match pair precedes P} \\
      \max_G LCS_k(G) + 1 & \text{ over all G preceding P}
    \end{array}
  \right .
\end{eqnarray}
In other words, we are looking for the longest chain of match pairs such that a pair which occurs earlier in the chain precedes the pairs which come later. The chain ending with a match pair $P$ can be constructed in two ways: (i) $P$ is added to some preceding shorter chain ending with $G$ or (ii) a new chain containing only $P$ is started.\\

\noindent Now that we have reformulated the relation for \LCSK{}, we show how to compute it efficiently. First we take a step back to the description of Hunt's algorithm for $LCS$. There we mentioned that the order in which the points within a row are processed matters since reads and updates to the helper array $M$ happen interchangeably. If we allow two traversals of the points, the first one can only read and the second one can only update $M$. Note that this does not have much effect on the result of the $LCS$ computation, but it only removes the restriction on the order in which we have to process the points within a row. The decoupling of the reads and the updates of $M$ is an idea which we use to generalize the algorithm to compute \LCSK{}. Namely, if we decouple the start and end points of all the match pairs and process them in row-major order, at every row we can: 1) do the reads of $M$ for all the start points and then 2) update $M$ with values of all the end points. This entire algorithm is summarized in Algorithm \ref{lcsk.fastalgorithm}.

\begin{algorithm}[h]
\begin{algorithmic}[1]  
\For{$0 \le i < m$}
  \ForAll{$x=(i,j) \in StartPointsForRow(i)$}
    \State $P\gets$ MatchPair(x)
    \State $LCS_{k,start}(P)\gets d$ s.t. $M_d < j \le M_{d+1}$
  \EndFor
  \ForAll{$x=(i,j) \in EndPointsForRow(i)$}
    \State $P\gets$ MatchPair(x)
    \State $LCS_{k,end}(P)\gets LCS_{k,start}(P)+1$
    \State $M_{LCS_{k,end}(P)}\gets \min(M_{LCS_{k,end}(P)}, j)$
  \EndFor
\EndFor

\State \textbf{return} $\max_P LCS_{k,end}(P)$
\end{algorithmic}
\caption{Computing $LCS_k$ by decoupling start and end points}
\label{lcsk.fastalgorithm}
\end{algorithm}

\noindent In the Algorithm \ref{lcsk.fastalgorithm} we use several quantities: 1) $LCS_{k,start}(P)$ stores the value read from $M$ at the start point of match pair $P$, 2) $LCS_{k,end}(P)$ stores the value of $LCS_k$ at the end point of $P$ and 3) $MatchPair(x)$ retrieves a match pair for which $x$ is a start or an end point. Lines 2-5 of the  algorithm can be implemented in two different ways: 1) following Hunt's paradigm and performing binary search over the $M$ array to do the reads and 2) following Kuo's paradigm and doing all the reads in one linear pass over array $M$. Since we can estimate the number of operations needed for both variants, at each row $i$ we dynamically choose between the two options. Doing this picks up the benefits of both approaches and makes our algorithm work efficiently in both sparse and dense rows.  $StartPointsForRow(i)$ and $EndPointsForRow(i)$ can be computed in $\mathcal{O}(n + m + r)$ time in multiple ways: 1) by using a suffix array based approach proposed by Deorowicz and Grabowski \cite{lcskfast} or 2) hash the $k$-mers of $B$ and create a hash table mapping to the indexes (if it happens that $\Sigma ^k$ is small enough to fit a computer word). Querying that table with $k$-mers from $A$ trivially gives the start/end points for a wanted row. For more details see Appendix \ref{generatingmatchpairs}.

\subsection{Complexity}
Generating all the match pairs takes $\mathcal{O}(m + n + r)$ time. Computing the update for row $i$ takes $\mathcal{O}(min(r_i\ log\ l, r_i + l))$ time. Summing over all the rows implies that
\begin{eqnarray}
\label{lcsk:complexity}
  \sum_{r=0}^{m-1} \min(r_i\ log\ l, r_i + l) \le \mathcal{O}(\min(r\ log\ l, r + ml))
\end{eqnarray}

\begin{theorem}
  \label{lcsk:complexity:thm}
The presented algorithm computes $LCS_k$ with the runtime complexity of $\mathcal{O}(m + n + r + \min(r\ log\ l, r+ml))$. The memory complexity is $\mathcal{O}(l + m + n)$.
\end{theorem}
\begin{proof}
The runtime complexity directly follows from adding up the complexities for generating the match pairs with the right side of Equation \ref{lcsk:complexity}. Regarding memory consumption, we need to $\mathcal{O}(l)$ memory for array $M$, and $\mathcal{O}(m + n)$ memory to generate the start/end points of the match pairs (same for both described approaches to generate them). This does not takes into account the memory needed for the reconstruction of the solution, which requires $\mathcal{O}(r)$ memory.

\end{proof}

\section{Algorithm to compute $LCS_{k+}$}
\label{section:lcskpp}

In this section, we show how to modify the algorithm for \LCSK{} to compute \LCSKP{}. Similar to \LCSK{}, we look at the \LCSKP{} computation as finding chains of match pairs where consecutive ones precede or continue (see Definition \ref{def:continuation}) one another. This makes it possible to achieve the runtime complexity of $\mathcal{O}(m + n + r + \min(r (\log l + k), r + ml))$ and the memory complexity of $\mathcal{O}(l + m + n)$.

\begin{definition}[Continuation of match pairs]
Let $P=(i_P,j_P)$ and $G=(i_G,j_G)$ be $k$-match pairs. Then $P$ \textbf{continues} $G$ if $i_P=i_G+1$ and $j_P=j_G+1$. $P$ is only one down-right position from $G$, see Figure \ref{fig:kmatchpairs}.
\label{def:continuation}
\end{definition}

\noindent We reproduce the relation for computing $LCS_{k+}$ over match pairs from
\cite{PaveticZS14}:

\begin{eqnarray} \label{lcskp:dp}
  dp(P) & = & \max
  \left \{
    \begin{array}{ll}
      k & \\
      dp(G) + 1 & \text{if P is a continuation of G} \\
      \max_G dp(G) + k & \text{over all G preceding P}
    \end{array}
  \right .
\end{eqnarray}

\noindent Note that this relation computes the actual number of characters in the $LCS_{k+}$, whereas the corresponding relation for $LCS_k$ computes the number of blocks of length $k$. This difference requires us to reconsider the compressed representation $M$ and make important changes. Instead of the array $M$, we introduce a new array $N$, defined as follows. Lets assume that we have finished processing the start and end points in row $i-1$. Let $N_d$ denote the minimum $j$ such that $dp(P) \ge d$ for some match pair $P$ with end point $(i', j)$. Furthermore, we let $N_0 = 0$ and $N_d = \infty$ when no such $j$ exists. Note that by using '$\ge$' instead of '$=$' we make $N$ non-decreasing. Indeed, since $LCS_{k+}(i-1, N_{d+1}) \ge d + 1 \ge d$, $N_d$ cannot be greater than $N_{d+1}$.\\

\noindent Using $N$ to compute the maximum over preceding pairs in Equation \ref{lcskp:dp} turns out to be the same as for $LCS_k$. However, updating $N$ needs an adjustment. If we look at the end $(i, j)$ of match pair $P$ and suppose that we have calculated the corresponding $dp(P)$, the first temptation is to simply set $N_{dp(P)} \gets \min(N_{dp(P)}, j)$. In order to maintain the non-decreasing property of $N$, we must ensure that $N_d \le j$ for all $d \le dp(P)$. However, it happens that we don't really need to update the whole prefix. If $dp(P)$ is computed from a preceding match pair $G$ as $dp(G) + k$, upon processing the start point $(i-k+1, j-k+1)$ of $P$, we have $N_{dp(G)} < j-k+1 \le j$. Since for any $d$, $N_d$ may only decrease as we move from row to row, this will also hold when we reach the end point of $P$ (in row $i$). Therefore, $N_{dp(G)} = N_{dp(P) - k}$ will already be smaller than or equal to $j$. The same holds for all indices less than $dp(G)$, as $N$ is kept non-decreasing so we only need to set $N_{dp(P) - s} \gets \min(N_{dp(P) - s}, j)$ for $s \in {0, .., k-1}$. If we encounter $N_{dp(P) - s} = j$ we stop, as further values of $N$ are already smaller than or equal to $j$. Since the pairs are sorted by column, this bounds the total time spent iterating through $N$ for a single row by $O(l)$. In the case $dp(P) = dp(G) + 1$ where $P$ is a continuation of $G$ we have $N_{dp(G)} \le j - 1 < j$, so it is enough to set $N_{dp(P)} \gets \min(N_{dp(P)}, j)$.\\

\begin{figure}[t]
\centering
\begin{tabular}{| c | c | c | c | c | c | c | c | c | c | c | c | c | c | c | c |}
\hline
i & $N_{0}$ & $N_{1}$ & $N_{2}$ & $N_{3}$ & $N_{4}$ & $N_{5}$ & $N_{6}$ & $N_{7}$ & $N_{8}$ & $N_{9}$ & $N_{10}$ & $N_{11}$ & $N_{12}$ & $N_{13}$ & $N_{14}$ \\
\hline
42 & 0 & 5 & 5 & 5 & 5 & 33 & 43 & 43 & 43 & 43 & 44 & 49 & 49 & 49 & 49 \\
\hline
45 & 0 & 5 & 5 & 5 & 5 & 31 & 31 & 31 & 31 & 43 & 44 & 49 & 49 & 49 & 49 \\
\hline
\end{tabular}

\caption{Example used to highlight the difference in the updates between \LCSK{} and \LCSKP{}. It shows a typical situation where $k$ entries of $N$ need to be changed. The table shows the state of $N$ after finishing with rows $42$ and $45$. Note that the content of the strings does not matter, we assume a situation where $k = 4$ and the only match pair in that range is $(42, 28)$ ($N$ does not change between rows $43$ and $44$). When processing the start point of $P = (42, 28)$ in row $42$, we find that $N_4 < 28 \le N_5$, so we take $d = 4$. This means that using $P$ we can extend a sequence of length $4$ into a sequence of length $8$. So when we reach row $45$, and see the end point $(45, 31)$ of $P$, we calculate $dp(P) = d + k = 4 + 4 = 8$. Finally, we set $N_5$ through $N_8$ to $j = 31$.}

\label{fig:lcskp:ex1}
\end{figure}

\subsection{Complexity}

In the sparse case, the time complexities of querying and updating $N$ are $\mathcal{O}(r_i \log l)$ and $\mathcal{O}(k r_i)$ respectively, or $\mathcal{O}(r_i (\log l + k))$ in total. For the dense case we get $\mathcal{O}(r_i + l)$ for both, yielding $\mathcal{O}(\min(r_i (\log l + k), r_i + l))$ as the time complexity of processing a single row. Summing over all the rows, this is bounded by $\mathcal{O}(\min(r (\log l + k), r + ml))$. The runtime complexity of the whole algorithm is presented in Theorem \ref{lcskp.theorem}.

\begin{theorem}
The presented algorithm computes $LCS_{k+}$ with the runtime complexity of $\mathcal{O}(m + n + r + \min(r (\log l + k), r + ml))$. The memory complexity is $\mathcal{O}(l + m + n)$.
\label{lcskp.theorem}
\end{theorem}
\begin{proof}
  The analysis of memory complexity is the same as in Theorem \ref{lcsk:complexity:thm}.
  For time complexity, we again add up the complexities of generating match pairs and calculating the length of the \LCSKP{}.
\end{proof}

\noindent Finally, we mention that one can achieve $\mathcal{O}(\log l)$ runtime complexity for querying $N$ by using a more involved data structure (see Appendix \ref{advancedlcskp}). However, we argue that the added complications are not worth it, since we are often in a setting where values of $k$ range in $\mathcal{O}(\log n)$. A good example of that are the DNA aligners operating on genomes having billions ($\log_2 10^9 \sim 30$) of nucleotides, with typical values of $k$ ranging from $10$ to $32$ \cite{SNAP}. Additionally, too large values of $k$ result in absence of match pairs, which does not make them useful \cite{PaveticZS14}. 

\section{Notes on implementation}
\label{section:reconstructionmemory}

We implemented the algorithms described in this paper and made the code available at \url{https://github.com/google/fast-simple-lcsk}. This section briefly describes some details of our implementation which haven't been addressed in the rest of the paper.\\

\noindent  The analysis of memory complexity for the reconstruction of the optimal solution in Section \ref{lcsk.algorithm} shows that maintaining the chains of match pairs causes the biggest part of the memory consumption - if we store the match pairs until the end of the computation to do the reconstruction we need $\mathcal{O}(r)$ memory. For long inputs, this makes it impossible to fit the computation into RAM of a single computer. To reduce the memory requirements we can observe the following: at any moment of the processing, we will have a set of reconstruction paths ending with a match pair contained in the array $M$ - we only need to keep the match pairs on these paths. As soon a there is no reconstruction path going from $M$ to some match pair $x$, we can delete $x$. This is implemented as follows: every match pair is reference counted and has a pointer to its predecessor in the reconstruction path. The last match pair in every reconstruction path is pointed to from array $M$. As soon as $M$ stops pointing to such match pair, its reference count drops to zero and it is deleted. This can further cause that the reference count of its predecessor dropped to zero so that one gets deallocated, etc.\footnote{In C++ this can be easily achieved by using std::shared\_ptr.} \\

\noindent In order to demonstrate savings of the memory consumption on the real world data, we performed simple experiments on different chromosomes from the human genome\footnote{\emph{Homo\_sapiens.GRCh38 (Release 88)} obtained from \emph{www.ensembl.org} \cite{doi:10.1093/nar/gkv1157}}. We computed \LCSK{} of several chromosomes with themselves and compared the number of match pairs with the maximum number of match pairs kept in memory during computation for different values of $k$, in order to get an estimate on their relation in real data. Figure \ref{fig:reconstruction:savings} shows that the savings of the memory consumption are significant. In particular, we can see that the described optimization can save \fortilde{~700-1000x} of the memory. The savings tend to increase as the number of match pairs increases.

\begin{figure}[t]
\centering
\begin{tabular}{| c | c | r | r | r |}
\hline
\ \ k\ \ &\ \ chromosome \ \ & \ \ match pairs\ \ &\ \ max in memory\ \ &\ \ compression factor\ \ \\
\hline
\hline
30 & 1 & 5 627 330 181 & 7 802 719 & 721.20\\
30 & 2 & 5 737 185 065 & 8 186 269 & 700.83\\
30 & 3 & 3 575 560 336 & 6 778 074 & 527.51\\
\hline
29 & 1 & 6 428 219 516 & 8 090 660 & 794.52\\
29 & 2 & 6 368 795 382 & 8 494 144 & 749.78\\
29 & 3 & 3 971 642 925 & 7 020 437 & 565.72\\
\hline
28 & 1 & 7 374 448 317 & 8 399 480 & 877.96\\
28 & 2 & 7 108 071 229 & 8 824 093 & 805.52\\
28 & 3 & 4 440 640 300 & 7 277 931 & 610.15\\
\hline
27 & 1 & 8 540 954 582 & 8 732 809 & 978.03\\
27 & 2 & 8 012 803 096 & 9 179 547 & 872.89\\
27 & 3 & 5 014 652 783 & 7 554 106 & 663.83\\
\hline
26 & 1 & 9 954 502 925 & 9 092 744 & 1094.77\\
26 & 2 & 9 100 424 727 & 9 537 800 & 954.14\\
26 & 3 & 5 708 852 882 & 7 849 184 & 727.31\\
\hline
\end{tabular}

\caption{The comparison of the number of all the match pairs and the number of max match pairs kept in memory at any moment of the computation for the first few chromosomes from the human genome and varying values of $k$. The ratios of these numbers directly translates to the savings in the memory required for the reconstruction. We note that in our case of computing \LCSK{} of a chromosome with itself, the number of match pairs is going to be a sum of squares.}

\label{fig:reconstruction:savings}
\end{figure}

\section{Conclusion and future work}

We have demonstrated a single and simple algorithm for computing both \LCSK{} and \LCSKP{} of two strings. The algorithm beats the runtime complexity of the existing approaches and does not require complex data structures. Also, we have demonstrated a heuristic for reducing the memory needed to reconstruct the solution. Experiments on real data demonstrated savings of \fortilde{~700-1000x}. Furthermore, we made our implementation widely accessible.
\\
\\
\noindent As a direction for future research we would like to pose few questions:
\begin{itemize}
\item[$\bullet$] Is it possible to create a link between the memory optimization we described with what is known as dominant points in $LCS$-related literature (see Appendix \ref{dominantpoints})?
\item[$\bullet$] It is not clear that all the match pairs are useful in the search for the optimal sequences. Can we speed up the algorithms by developing rules for discarding some of them, before even starting the computation?
\end{itemize}

\subsubsection{Acknowledgments.}
The authors would like to thank Maria Brbi\'{c} and Mario Lu\v{c}i\'{c} for the valuable comments on the manuscript.

\bibliographystyle{splncs03}             
\bibliography{fast-lcsk}            

\appendix

\newpage

\section{Appendix}
\subsection{Generating the match pairs}
\label{generatingmatchpairs}

Algorithm \ref{lcsk.fastalgorithm} assumes that it has a list of start and end points already available. Here we address how to obtain it. We offer two approaches which have the same runtime complexity, but different tradeoffs between the simplicity and the assumptions on the input data.\\

\noindent Deorowicz and Grabowski \cite{lcskfast} described an algorithm for enumerating all the match pairs by using a suffix array built over the string $B\#A$\footnote{$B\#A$ = string B, concatenated with character '\#', concatenated with string A}. The $LCP$ (Longest Common Prefix) table is used to group all the suffixes of that string sharing a prefix of at least $k$. With careful bookkeeping it is possible to enumerate all the columns $j$ which correspond to match pairs starting in row $i$. For more details please consult the referenced paper.\\

\noindent We note that in practice we are often in a setting where either the alphabet size $\Sigma$ is small or useful values of $k$ are small (e. g. DNA has $\Sigma=4$ and the popular tools for aligning the DNA often set $k$ to a range 10-32 \cite{SNAP}). If it happens that $\Sigma ^k$ is small enough to fit a 64-bit integer, we can easily and cheaply obtain perfect hashing of the $k$-mers by treating them as $k$-digit number in base $\Sigma$. Having that, we can build a hash table $\mathcal{H}$ mapping from the hashes of all the $k$-mers of $B$ to the indices of their start position. Enumerating all the match pairs starting in a row $i$ then comes down to hashing $A[i:i+k)$ and looking up all the indices from $\mathcal{H}$. This gives us a simpler algorithm, but still relevant in practice.\\

\noindent Both of the described approaches show how to generate the match pairs in $\mathcal{O}(n + m + r)$ time. We note that Algorithm \ref{lcsk.fastalgorithm} requires the start and end match points at every row. The start points for row $i$ are obtained directly by generating the match pairs at row $i$. The end points for row $i$ are obtained by generating the match pairs at row $i-k+1$.

\subsection {An alternative to $\mathcal{O}(k)$ updates for \LCSKP{}}
\label{advancedlcskp}

The operations we need are querying a single element of the array,
and setting all elements in a given prefix $[0..i]$ to the minimum
of a given value $v$ and their present value.\\

\noindent One approach is to use a complete binary tree in which the leaves correspond to array elements in order. An internal node then represents an interval of the array. To simplify things, we think of leaves as intervals consisting of a single element. Every node of the tree has an associated value, initially set to $\infty$. We denote the leaf representing element $N_i$ by $L(i)$. The algorithm is as follows:

\begin{itemize}
\item[$\bullet$] The query for $N_i$ is done by finding the minimum of the values
along the path from root to the appropriate leaf $L(i)$.
\item[$\bullet$] The update for prefix $[0, i]$ with value $v$ is done by setting
the values left siblings (if they exist) of nodes along the path from root to $L(i + 1)$
to the minimum of $v$ and their present value.
\end{itemize}

\noindent Looking at the query and update together, we see that querying for $i$ we will, for every previously applied update on $[0 .. j]$ such that $i \le j$, encounter exactly one node that was affected by it, and thus correctly calculate the current value of $N_i$. Since the depth of the tree is $\mathcal{O}(\log l)$, that is also the time complexity of our operations.\\

\noindent In total this gives the runtime complexity of $\mathcal{O}(m + n + r + \min(r \log l, r + ml))$.

\subsection{Dominant points}
\label{dominantpoints}

\noindent The example strings shown in Figure \ref{fig:reconstruction:memory} (left) generate match pairs at almost all the indexes. In this case $r=\Omega(mn)$, which makes the memory required to reconstruct the solution extremely high. Still, looking at the table, we can observe the following: in order to build any of the table entries with value 4, we only need to keep the topmost-leftmost entry with value 2, as opposed to keeping all of them. In fact, the points that we need to keep in order to reconstruct the solution are known in the $LCS$-related literature as \emph{dominant points}.

\begin{definition}[Dominant points \cite{Apostolico1987,Apostolico:1992:FLC:135853.135858}]
Lets look at a table $T(i,j)=$\LCSK{}$(i,j)$ or $T(i,j)=$\LCSKP{}$(i,j)$. A point $(i, j)$ is then called $q$-dominant if $T(i,j)=q$ and for any other $(i', j')$ such that $T(i',j')=q$ it holds that either ($i'>i$ and $j'\le j$) or ($i'\le i$ and $j'>j$) is true. The dominant points are then the union over $q$-dominant points over all different $q$.
\end{definition}

\begin{figure}[t]
\centering
\begin{tikzpicture}[scale=0.4]
\input{reconstruction_memory_lcsk_a}
\end{tikzpicture}
~
\begin{tikzpicture}[scale=0.4]
\input{reconstruction_memory_lcsk_b}
\end{tikzpicture}

\caption{The full $LCS_k$ table for two different string pairs, with $k = 2$. The encircled fields are called dominant points. For the purposes of reconstructing the solution, it is sufficient to only keep the match pairs ending at these locations.}
\label{fig:reconstruction:memory}
\end{figure}

\end{document}